\documentclass[12pt]{article}
\usepackage{amsmath,amsfonts,amsthm,amssymb}

\usepackage{fancybox}
\usepackage{amstext}

\usepackage{setspace}  

\newcommand{\version}{September 22, 2012}

\swapnumbers
\font\notefont=cmsl8

\theoremstyle{plain}
\newtheorem{thm}{THEOREM}[section]
\newtheorem{lm}[thm]{LEMMA}
\newtheorem{cl}[thm]{COROLLARY}

\theoremstyle{definition}

\theoremstyle{remark}
\newtheorem{remark}[thm]{REMARK}
\newcommand{\upchi}{\raise1pt\hbox{$\chi$}}
\newcommand{\R}{{\mathord{\mathbb R}}}
\newcommand{\C}{{\mathord{\mathbb C}}}
\newcommand{\Z}{{\mathord{\mathbb Z}}}

\newcommand{\rS}{{\mathord{\mathbb S}}}

\newcommand{\cH}{{\mathord{\cal H}}}

\newcommand{\Tr}{{\mathord{\rm Tr}}}

\newcommand{\slangle}[1]{{_{_{#1}}}\mspace{-3mu}\langle}
\newcommand{\srangle}[1]{\rangle\mspace{-4mu}{_{_{#1}}}}
\begin{document}
\title{Proof of an entropy conjecture for Bloch coherent spin states
  and its generalizations\let\thefootnote\relax\footnotetext{\copyright 2012 \ by the authors. This article may be
      reproduced in its entirety for non-commercial purposes.}} %
\author{
  \begin{tabular}{ccc}
    Elliott H. Lieb &\hspace{2cm}& Jan Philip
    Solovej\\
    \normalsize Departments of Physics and Mathematics && 
    \normalsize Department of Mathematics\\ 
    \normalsize Jadwin Hall, Princeton University && \normalsize University of Copenhagen\\
    \normalsize Washington Road&&\normalsize Universitetsparken 5\\
    \normalsize Princeton, N.J. 08542  & &
    \normalsize DK-2100 Copenhagen, Denmark\\
    \normalsize {\it e-mail\/}: lieb@princeton.edu &&
    \normalsize {\it e-mail\/}: solovej@math.ku.dk
  \end{tabular}
  \bigskip
  \date{\version}} 

                                \maketitle

\begin{abstract}
  Wehrl used Glauber coherent states to define a map from quantum
  density matrices to classical phase space densities and conjectured
  that for Glauber coherent states the mininimum classical entropy
  would occur for density matrices equal to projectors onto coherent
  states. This was proved by Lieb in 1978 who also extended the
  conjecture to Bloch $SU(2)$ spin-coherent states for every angular
  momentum $J$. This conjecture is proved here. We also recall our
  1991 extension of the Wehrl map to a quantum channel from $J$ to
  $K=J+1/2, \, J+1,\, \ldots$, with $K=\infty$ corresponding to the
  Wehrl map to classical densities. For each $J$ and $J<K\leq \infty$
  we show that the minimal output entropy for these channels occurs
  for a $J$ coherent state. We also show that coherent states both
  Glauber and Bloch minimize any concave functional, not just entropy.

\end{abstract}

\vfill\eject
\section{Introduction}

Coherent states in the Hilbert space $\cH=L^2(\R^n)$, which originate in 
the work of Schr\"odinger, Bargmann, Glauber and others, are certain normalized Gaussian functions
parametrized by points in classical phase space, $(p,q)  \in \R^n\times \R^n$. They are denoted by
$|p,q \rangle$. See, e.g., \cite{L2} or Section~\ref{sec:basic} for definitions. 

Given a density matrix $\rho$ on $\cH$ (i.e., $\rho $ is a positive
semi-definite operator with trace $\Tr \rho =1$), its von Neumann
entropy is $S(\rho) =-\Tr \rho \ln \rho$. This is always non-negative,
but the usual classical Boltzmann type density (e.g.,
$f(p,q)=Z^{-1}\exp -\beta (p^2 +V(q))$ can have an arbitrarily
negative entropy $-\int_{\R^n\times \R^n}f(p,q) \, \ln f(p,q) dp
dq$. To rememdy this, and other problems with the classical
approximation to entropy, A. Wehrl \cite{WA} used coherent states to
propose another definition, as follows:

Define $\rho^{\rm cl}(p,q) := \langle p,q| \, \rho \, |p,q\rangle $. This
function has several names.  One is the Husimi $Q$-function
\cite{Hu}. Berezin~\cite{Be} called it the covariant symbol. The name
we shall use here is {\it the lower symbol} of $\rho$ (see \cite{L1}).
Since $(2\pi)^{-n} \int_{\R^n\times \R^n} |p,q\rangle \langle p,q| dp
dq$ is the unit operator on $\cH$, we see that $\rho^{\rm cl}(p,q)$ is a
probability density with respect to the classical measure $(2\pi)^{-n}
dp dq$.  In \cite{Be,L1} it is shown that the trace $\Tr(f(\rho))$ of a
convex function $f$ of $\rho$ is bounded below by the corresponding
classical integral $(2\pi)^{-n} \int_{\R^n\times \R^n}
f(\rho^{\rm cl}(p,q))dpdq$. Together with the corresponding upper bound
for what is called the upper symbol (or contravariant symbol in
\cite{Be}) they are often referred to as the Berezin-Lieb
inequalities. The inequalities are, of course, reversed for concave
functions. In this paper we shall not be concerned with the upper
symbol. 

Wehrl uses the lower symbol, to define 
\begin{equation}\label{eq:wehrl}
S^{\rm cl}(\rho) = - (2\pi)^{-n} \int_{\R^n\times \R^n}\rho^{\rm
  cl}(p,q) \, \ln \rho^{\rm cl}(p,q) dp dq
\end{equation} 
Since $0\leq\rho^{\rm cl}(p,q)\leq 1$, we see that $S^{\rm cl}(\rho) \geq 0$, as desired. 

A question raised by Wehrl is what $\rho $ minimizes $S^{\rm
  cl}(\rho)$, and he conjectured that this occurs exactly when $\rho$
is a one-dimensional projection onto any coherent state, i.e., $\rho =
|p,q\rangle \langle p,q|$ for any choice of $p,q$.  This conjecture
was proved in \cite{L2}. Later Carlen \cite{Ca} gave a different proof
based on the log-Sobolev inequality; this proof included the
uniqueness statement for the first time. A decade later Luo \cite{Lu}
gave a proof based on hypercontractivity, which is
closely related to the log-Sobolev inequality. 

In \cite{L2} a similar conjecture was made for the coherent states in
the Hilbert spaces of the irreducible representations of
$SU(2)$. These {\it angular momentum coherent states} are very useful
for the physics of quantum spin systems. They are usually called Bloch
coherent states \cite{Bl,Ra} and will be the subject of this
paper. {\it We shall prove the conjecture in \cite{L2} that the analog
  of Wehrl's conjecture holds here as well}, i.e., that the classical
entropy is minimized by coherent states. We do not prove that
coherent states alone minimize the entropy, however.

Previously, the conjecture had been established only for a few
low-dimensional representations of $SU(2)$. The 2-dimensional spin $1/2$ case is
simple as all pure states are coherent states. This was already
pointed out in \cite{L2}. For the 3-dimensional case of spin $1$ the
conjecture was solved by Scutaru \cite{Scu} and by Schupp
\cite{Sch} who also solved it for the 4-dimensional representation
corresponding to spin $3/2$. Bodmann \cite{Bo} proved a lower bound on
the classical entropy which is asymptotically correct for large spin
$J$.

We prove more than that $S^{\rm cl}(\rho)$ is minimized when $\rho$ is
a projection onto a coherent state. We prove this for {\it all concave
functions} $f(t)$, not just $f(t)=-t\ln(t)$. In fact, the original
proof in \cite{L2}
for the Gaussian Glauber states was for $f(t)=-t^p$, $p\geq1$, and the
proof for $-t\ln(t)$ followed by taking the limit $p\to1$.  
The extension to $f(t)=t^p$ for $0<p<1$ was given by Carlen
\cite{Ca}.  To our knowledge it has not been proved for
general concave functions. In Theorem~\ref{thm:glauber} we show
exactly that by approximating Glauber coherent states by
Bloch coherent states in an appropriate limit of large
spin. The particulars of this approximation are in the appendix. 

In order to prove the conjecture for all spin we utilize a
generalization of coherent states, called {\it coherent operators}
introduced by us in \cite{LiSo}. These are operators that map density
matrices in one $SU(2)$ space, characterized by an angular momentum
(or "spin") $J$ to a density matrix in a spin $K$ space. These maps
are, in fact, {\it quantum channels}, i.e., completely positive trace
preserving maps, as will be made clear in Lemma~\ref{lm:Phi2nd}.  

The Bloch coherent states map density matrices from $J$ to functions
on the classical phase space, i.e., the $2$-sphere $\rS^2$
(which can be thought of as $K= \infty$ \cite{L1}) via the lower
symbol map mentioned above. 
We do this in small steps, so to speak, by going from $J$ to $J+1/2$
to $J+1,\ldots$.  For each finite $K$ on the way we prove that
projections onto coherent states in $J$ minimize the von Neumann
entropy of the lifted density matrix in $K$. In other words we
determine the {\it minimal output entropy} of the quantum coherent
operator channels to be the entropy of the output of a coherent
state. We then show that after an appropriate scaling, the
limit $K\to \infty$ gives us the desired classical (lower symbol)
entropy, and thus we prove the conjecture for the entropy and for any concave
function. 

An important observation in our procedure is to note that the quantum
coherent operator channels have a simple expression in terms of
bosonic second quantization, i.e., bosonic creation and annihilation
operators. 

Coherent states can be generalized to any compact semisimple Lie
group, not just $SU(2)$ (see \cite{Per,Si}), and there we expect that a
similar result holds.

In Section~\ref{sec:basic} we define Bloch coherent states, lower
symbols, and discuss the corresponding Berezin-Lieb inequality. We
also introduce the quantum coherent operator channels from $J$ to
$K$. (While we are interested here in $K>J$, the map is also defined
when $K<J$.)  In Section~\ref{sec:formula} we derive a more explicit
formula for the quantum coherent operator channel which allows us, in
Section~\ref{sec:bose}, to give a bosonic second quantization
representation of the channels. In fact, Section~\ref{sec:formula} is
not important for our main conclusion as we could have defined the
quantum channels from the second quantization formulation in
Section~\ref{sec:bose}. We include Section~\ref{sec:formula} in order
to connect to our previous work in \cite{LiSo}. Following that, we
show, in Section~\ref{sec:main}, that coherent states in $J$ minimize
the output von Neumann entropy in $K$ or more generally the trace of
any concave function. In the last Section~\ref{sec:classicallimit} we
study the classical limit $K\to \infty$ and use the Berezin-Lieb
inequality to prove the conjecture on the classical entropy.

{\bf Acknowledgment:} Part of this work was carried out at the
Mathematics Department of the Technical University, Berlin and at the
Newton Institute, Cambridge. We are grateful to both and, in particular,
to Ruedi Seiler for hosting our stay in Berlin.

Thanks go to Eric Carlen, Rupert Frank, and Peter Schupp for their
helpful comments on a preliminary version of the manuscript and to Anna
Vershynina for a careful reading of the manuscript.

This work was partially supported by U.S. National Science
Foundation grant PHY-0965859 (Elliott Lieb ), by a grant from the Danish
research council (Jan Philip Solovej) and by a grant from the Simons
Foundation (\#230207 to Elliott Lieb).

\section{Basic definitions and main results} \label{sec:basic}

For all integer or half-integer $J$ we let $\cH_J$ denote the spin $J$
representation space of $SU(2)$, i.e., $\cH_J=\C^{2J+1}$. The
corresponding classical phase space is $\rS^2$, the unit sphere in
$\R^3$. For each point $\omega\in\rS^2$ we have the one-dimensional
coherent state projection
$P_\omega^J=|\omega\srangle{J}\slangle{J}\omega|$ projecting $\cH_J$
onto the subspace of maximal spin in the direction $\omega$, i.e., the
one-dimesional subspace of $\cH_J$ corresponding to the eigenspace of
$\omega\cdot {\bf S}_J$ with eigenvalue $J$. Here ${\bf S}_J$ is the
vector of spin operators, i.e, the representation on $\cH_J$ of the
standard generators ${\bf S}=(S_x,S_y,S_z)$ of $SU(2)$.  The vector
$|\omega\srangle{J}$ is only defined up to a phase, but this will not
play a role here as only the projection $P_\omega^J$ is important.  We
will use the notation that $\uparrow,\downarrow\in\rS^2$ are
respectively the north and south pole.

The coherent state transform is based on the identity that 
\begin{equation}\label{eq:identityres}
\frac{2J+1}{4\pi}\int_{\rS^2} |\omega\srangle{J}\slangle{J}\omega|
d\omega=\frac{2J+1}{4\pi}\int_{\rS^2} P_\omega^J
d\omega=I_J,
\end{equation}
where $I_J$ is the identity on $\cH_J$. 
If $\rho$ is a density matrix on $\cH_J$ its {\it lower symbol} is the
function on
$\rS^2$ given by
\begin{equation}
\Phi^\infty(\rho)(\omega)=\slangle{J}\omega|\rho|\omega\srangle{J}=\Tr_J(P_\omega^J\rho),
\end{equation}
where $\Tr_J$ is the trace on $\cH_J$.  The classical entropy of
$\rho$ is
$$
S^{\rm cl}(\rho)=-\frac{2J+1}{4\pi}\int_{\rS^2} \Phi^\infty(\rho)(\omega)\ln(\Phi^\infty(\rho)(\omega))d\omega.
$$
We are using the notation $\Phi^\infty$ for the lower symbol since we
shall consider it as the natural classical limit $k\to\infty $ of the quantum
channels $\Phi^k$  defined below.

The Berezin-Lieb \cite{Be,L1} inequality for the lower symbol states
that for any concave function $f:[0,1]\to\R$ we have 
\begin{equation}\label{eq:BL}
\Tr_J f(\rho)\leq \frac{2J+1}{4\pi}\int_{\rS^2} f(\Phi^\infty(\rho)(\omega))d\omega.
\end{equation}
 The inequality follows from
(\ref{eq:identityres}) as a consequence of Jensen's inequality. 

The conjecture from \cite{L2} that we shall prove here is that $S^{\rm
  cl}$ is minimized when the density matrix is any coherent state
projection, e.g.,  $\rho=|\uparrow\srangle{J}\slangle{J}\uparrow|$. In
this case the lower symbol is
$\Phi^\infty(|\uparrow\srangle{J}\slangle{J}\uparrow|)(\omega)
=|\slangle{J}\omega|\uparrow\srangle{J}|^2$.
In fact, we shall prove the more general statement that the same is
true if the function $-t\ln(t)$ is replaced by any concave
funtion. Our main theorem is the following.
\begin{thm}[Lower symbols of Bloch coherent states minimize concave averages]\label{thm:conjecture}
Let $f:[0,1]\to\R$ be a concave function\footnote{\label{footnote} It is, in fact,
  enough to assume that $f:[0,1)\to\R$, i.e., to allow that
  $\lim_{t\to0-}f(t)=-\infty$. Only coherent state projections have
  lower symbols that attain the value 1. If $\rho$ is not a coherent
  state projection
  we can find a concave function $\widetilde f:[0,1]\to\R$ such that
  $\widetilde f\geq f$ and $\widetilde
  f(\slangle{J}\omega|\rho|\omega\srangle{J})=f(\slangle{J}\omega|\rho|\omega\srangle{J})$.}. Then for any
density matrix $\rho$ on $\cH_J$ we have 
\begin{equation}\label{eq:conjecture}
\frac{2J+1}{4\pi}\int_{\rS^2}f(\slangle{J}\omega|\rho|\omega\srangle{J}) d\omega
\geq \frac{2J+1}{4\pi}\int_{\rS^2}f(|\slangle{J}\omega|\uparrow\srangle{J}|^2)d\omega.
\end{equation}
By  $SU(2)$ invariance $\uparrow$ could be replaced by any other point on
$\rS^2$. 
\end{thm}
The following analogous result for the Glauber coherent states is
proved by an
easy limiting argument which we give in the appendix. 
\begin{thm}[Lower symbols of Glauber coherent states minimize concave
  averages]\label{thm:glauber}
Let $f:[0,1]\to\R$ be a continuous concave function with  $f(0)=0$. Then for any
density matrix $\rho$ on $L^2(\R^n)$ we have 
\begin{equation}\label{eq:genwehrl}
(2\pi)^{-n} \int_{\R^n\times \R^n}f(\rho^{\rm
  cl}(p,q) )dp dq \geq (2\pi)^{-n} \int_{\R^n\times \R^n}f(|\langle p_0,q_0|p,q\rangle|^2)dp dq 
\end{equation}
for all $p_0,q_0\in \R^n$. 
\end{thm}
\begin{remark}As in the main Theorem~\ref{thm:conjecture} we could have allowed
$\lim_{t\to1-}f(t)=-\infty$ (see Footnote~\ref{footnote}).  We could, in fact, also allow
$f(0)\ne0$, but in this case the integrals on both sides of
(\ref{eq:genwehrl}) are either both $+\infty$ or both $-\infty$. Even
if $f(0)=0$ the integrals may still be $+\infty$, but the inequality
holds in the sense that either both sides are $+\infty$ or the right
side is finite.
\end{remark}
Using the fact that the Glauber coherent state $|p,q\rangle\in L^2(\R^n)$ is
explicitly given by
$$
\pi^{-n/4}\exp(-(x-q)^2/2+ipx),
$$
we have
$$
\langle p,q|\psi\rangle=\pi^{-n/4}\int \exp(-(x-q)^2/2-ipx)\psi(x)dx,
$$
for $\psi\in L^2(\R^n)$. 
The inequality
(\ref{eq:genwehrl}) for the rank one state $|\psi\rangle\langle\psi|$
then states that
$$
\int_{\R^n\times \R^n}f(|\langle p,q|\psi\rangle|^2)dp dq 
$$
is minimized for concave $f$ when $\psi$ is a Glauber coherent state.

We now define the quantum coherent operator channels. We refer to
\cite{LiSo} for details. 
For fixed $K$ and $J$ we let $P_{-}$  be the projection in 
$\cH_J\otimes\cH_K$ onto the minimal total spin $|K-J|$, i.e., onto the
unique copy of $\cH_{|K-J|}\subseteq \cH_J\otimes\cH_K$ on which the
tensor product representation acts irreducibly\footnote{Strictly
  speaking the isometric imbedding of $\cH_{|K-J|}$ into
  $\cH_J\otimes\cH_K$ is given uniquely only up to a phase.}.
For simplicity we omit in our notation the dependence of
$P_-$ on $K$ and $J$. 

In the language of elementary quantum mechanics a particle of angular
momentum $K$ and one of angular momentum $J$ can combine in exactly
one way to produce a composite particle of angular momentum
$|K-J|$. The Hilbert space of this composite particle is the subspace 
$\cH_{|K-J|}\subseteq \cH_J\otimes\cH_K$.

If we let $k=2K-2J\in\Z$ we consider the map
$\Phi^k$ from operators on $\cH_J$ 
to operators on $\cH_K$ defined by the partial trace 
\begin{equation}
\Phi^{k}(\rho)=
\frac{2J+1}{2|K-J|+1}\Tr_J(P_-(\rho\otimes I_K)).
\end{equation}
This is a trace preserving completely positive map (see also
(\ref{eq:tildePhi}) and Lemma~\ref{lm:Phi2nd}), i.e., using the
language of quantum information theory it is a {\it quantum
  channel}. The trace preserving property is easily seen since the
partial traces $\Tr_J P_-$ and $\Tr_KP_-$ are both proportional to the
identities. In particular, $\Phi^{k}$ maps density matrices to density
matrices.  In the notation we have for simplicity omitted the
dependence of $\Phi^k$ on $K$ and $J$ and only kept the dependence on
the difference in dimension $k=2K-2J$.

Our main result about these channels is that they are majorized by
coherent states in the following sense. 
\begin{thm}[Coherent states majorize $\Phi^k$]\label{thm:Phimajor}
  For a density matrix $\rho$ on $\cH_J$ and $k=2(K-J)$ the sequence
  of eigenvalues of the density matrix $\Phi^k(\rho)$ is majorized by
  the sequence of eigenvalues of
  $\Phi^k(|\omega\rangle\langle\omega|)$, which by $SU(2)$ invariance
  is independent of $\omega\in\rS^2$.
\end{thm}
To say that a finite real sequence $a_1\geq a_2\geq \cdots\geq a_M$ 
majorizes another real sequence $b_1\geq b_2\geq \cdots\geq b_M$, written 
$(a_1,\ldots,a_M)\succ (b_1,\ldots, b_M)$ means that
\begin{align}
a_1\geq& b_1\nonumber\\a_1+a_2\geq& b_1+b_2\nonumber\\\vdots&\nonumber\\ a_1+\ldots+a_{M-1}\geq&
b_1+\ldots+b_{M-1}\nonumber\\ a_1+\ldots+a_{M}=&
b_1+\ldots+b_{M}.\label{eq:majorization}
\end{align}
Note the equality in the last condition
(\ref{eq:majorization}). 

It is a fact
that $(a_1,\ldots,a_M)\succ (b_1,\ldots,b_M)$ if and only if 
$$
\sum_{j=1}^Mf(a_j)\leq \sum_{j=1}^Mf(b_j)
$$
for all concave functions $f:\R\to\R$. This is often called Karamata's
Theorem,\footnote{In \cite{LS} the concave function is assumed to be
  monotone increasing. With the assumption of equality in
  (\ref{eq:majorization}) the assumption of monotonicity is not
  required.} cf.\ \cite{LS} Remark 4.7 after eq.\ (4.5.4). It is, in fact, enough that the
concave function is defined on the interval $[a_M,a_1]$.

If $A$ and $B$ are two Hermitian matrices of the same size we write $A\succ B $ if the
eigenvalue sequence of $A$ majorizes the eigenvalue sequence of
$B$. The notion of majorization of sequences can be easily generalized
to infinite summable sequences and trace class operators but we will
not need this here. 

It is an easy exercise, using the variational principle, to prove that if
$A,B,C$ are Hermitian matrices such that $A\succ B$ and $A\succ C$ then 
\begin{equation}\label{eq:variational}
A\succ \lambda B+(1-\lambda) C,
\end{equation}
for all $0\leq\lambda\leq 1$. 

As a consequence it follows from Theorem~\ref{thm:Phimajor} that 
the minimal output (von Neumann) entropy of the channel $\Phi^k$
is achieved for a coherent state, i.e.,
$$
\min_\rho S(\Phi^k(\rho))=
S(\Phi^k(|\omega\rangle\langle\omega|)).
$$ 
More generally, the output of
coherent states minmize the trace of concave functions. 
\begin{cl}[Minimization of the trace of concave functions]\label{cl:main}
If $f:[0,1]\to \R$ is concave and $\rho$ is a density matrix on
$\cH_J$ and $k=2(K-J)$ then 
$$
\Tr_K[f(\Phi^k(\rho))]\geq \Tr_K[f(\Phi^k(|\omega\rangle\langle\omega|))]
$$
for all $\omega\in\rS^2$.
\end{cl}
Of course, the inequalities about concave functions are reversed for
convex functions. 
\section{A formula for $P_-$}\label{sec:formula}
Our goal here is to find an explicit formula for the projection $P_-$
that projects $\cH_J\otimes\cH_K$ onto the subspace $\cH_{K-J}$,
under the assumption that $K\geq J$.

We start by choosing the standard preferred basis 
$$
|M\srangle{L},\quad  M=-L,\ldots,L,
$$
in $\cH_L$, in which $S_z$ is diagonal and $S_x$ is real. This specifies the 
basis up to an over-all phase.
We introduce the {\it anti}-unitary map $U_L:\cH_L\to\cH_L$
given by 
$$
U_L\sum_{M=-L}^L\alpha_M|M\srangle{L}=
\sum_{M=-L}^L(-1)^{L-M}\overline{\alpha_M}|-M\srangle{L}.
$$
This map has the property that 
\begin{equation}\label{eq:U_Laction}
U_L^{-1}{\bf S}_LU_L=-{\bf S}_L
\end{equation}
It is the unitary $e^{i\pi S_y}$ followed by complex conjugation in
the preferred basis.  Any anti-unitary satisfying (\ref{eq:U_Laction})
agrees with $U_L$ up to an over-all phase. Note that $U_L^{-1}=
(-1)^{2L}U_L$ and hence 
\begin{equation}\label{eq:UL}
  \langle U_L\phi|\psi\rangle=\langle U_L\psi|U_LU_L\phi\rangle
  =(-1)^{2L}\langle U_L\psi|\phi\rangle,
\end{equation}
for all $\phi,\psi\in\cH_L$. 

If $K\geq J$ then $\cH_K\subseteq \cH_{K-J}\otimes\cH_J$ 
(recall that $\cH_{K-J}\subseteq\cH_J\otimes\cH_K$ as we said in the beginning). 
We thus have a
sesqui-linear map
$$
\cH_J\times\cH_K\ni(\psi,\phi)\mapsto\slangle{J}\psi\|\phi\srangle{K}
\in\cH_{K-J},
$$
where the {\it partial} inner product
$\slangle{J}\psi\|\phi\srangle{K}$ is defined by the inner product in
$\cH_{K-J}$ as follows
$$
\slangle{K-J}\eta |
  \slangle{J}\psi\|\phi\srangle{K}\srangle{K-J}
=\slangle{(K-J)\otimes J}\eta\otimes \psi|\phi\srangle{(K-J)\otimes J}
$$
for all $\eta\in\cH_{K-J}$, where the last inner product is in
$\cH_{K-J}\otimes\cH_J$.
\begin{thm}[Formula for $P_-$]
If $K\geq J$ then we have for all $\psi\in\cH_J$ and $\phi\in\cH_K$ 
\begin{equation}\label{eq:P-formula}
P_-(\psi\otimes\phi)=\mu \
\slangle{J} U_J\psi\|\phi\srangle{K},
\end{equation}
where $\mu\in\C$ satisfies 
$$
|\mu|^2=\frac{2(K-J)+1}{2K+1}.
$$ 
\end{thm}
\begin{proof} Formula (\ref{eq:P-formula}) above for $P_-(\psi\otimes\phi)$ is clearly a
  bilinear map in $\psi$ and $\phi$. It is thus enough to prove the
  above formula for a linear spanning set for $\phi$ and $\psi$. Such
  spanning sets are provided by the coherent states $|\omega\srangle{J}$
  and $|\omega\srangle{K}$ for $\omega\in\rS^2$.
  It is thus enough to prove the formula for $\psi=|\omega'\srangle{J}$ and 
  $\phi=|\omega\srangle{K}$. For simplicity we will write $U_J|\omega'\srangle{J}
  =|U_J\omega'\srangle{J}$, where $U_J$ is the anti-unitary defined above.

  In the following we let ${\bf S}_J$ and ${\bf S}_K$, respectively,
  denote the spin operators on $\cH_J$ and $\cH_K$, respectively, and
  we let ${\bf S}$ be the total spin operator on $\cH_J\otimes\cH_K$,
  i.e., ${\bf S}={\bf S}_J\otimes I_K+I_J\otimes{\bf S}_K$. Since
  $\omega\cdot{\bf S}_K|\omega\srangle{K}=K|\omega\srangle{K}$, and
  since $\eta=|\omega'\srangle{J} -|U_J\omega\srangle{J}\slangle{J}
  U_J\omega|\omega'\srangle{J}$ has no component in the subspace
  $\omega\cdot{\bf S}_J=-J$, it is clear that
  $P_-\eta\otimes|\omega\srangle{K}=0$, for otherwise the total
  $\omega\cdot{\bf S}$ component of this vector would be bigger than
  the maximal possible namely $K-J$. Hence
  \begin{equation}\label{eq:P_1}
    P_-|\omega'\srangle{J}\otimes|\omega\srangle{K}=
    \slangle{J} U_J\omega|\omega'\srangle{J}P_-
    |U_J\omega\srangle{J}\otimes|\omega\srangle{K}.
  \end{equation}
  We have $\omega\cdot{\bf S}|U_J\omega\srangle{J}\otimes|\omega\srangle{K}
  =(K-J)|U_J\omega\srangle{J}\otimes|\omega\srangle{K}$ and thus, since 
  $P_-$ commutes with $\omega\cdot{\bf S}$,
  \begin{equation}\label{eq:P_2}
    P_- |U_J\omega\srangle{J}\otimes|\omega\srangle{K}= \mu'\
    |\omega\srangle{K-J}=\mu'\ \slangle{J}
    \omega\|\omega\srangle{K}
  \end{equation}
  for some complex scalar $\mu'$. Here we have used that $\slangle{J}
  \omega\|\omega\srangle{K}=|\omega\srangle{K-J}$ in
  $\cH_{K-J}\otimes\cH_J$, since
  $|\omega\srangle{K}=|\omega\srangle{K-J}\otimes|\omega\srangle{J}$. 
  Inserting (\ref{eq:P_2}) into (\ref{eq:P_1}) 
  we obtain
  $$
  P_-|\omega'\srangle{J}\otimes|\omega\srangle{K}=\mu'\
    \slangle{J} U_J\omega|\omega'\srangle{J}
    \slangle{J}\omega\|\omega\srangle{K}.
  $$
  Since $U_J^2=(-1)^{2J}$ and $U_J$ is anti-unitary
  we have $U_J(|\omega'\srangle{J}-\eta)=(-1)^{2J}
  \slangle{J}\omega'|U_J\omega\srangle{J}|\omega\srangle{J}$.
  Moreover, $U_J\eta$ has no component in the space $\omega\cdot{\bf
    S}_J=J$, hence $\slangle{J}U_J\eta\|\omega\srangle{K}=0$ and thus 
  $$
  P_-|\omega'\srangle{J}\otimes|\omega\srangle{K}=\mu'\ \slangle{J}
  U_J\omega|\omega'\srangle{J}\slangle{J}\omega\|\omega\srangle{K}=\mu\
  \slangle{J}U_J\omega'\|\omega\srangle{K},
  $$
  with $\mu=(-1)^{2J}\mu'$, which is what we wanted to prove. 
  
  We can find the modulus of $\mu$ from the fact that $\Phi^{-k}$,
  with $k=2(K-J)$ is trace preserving
  \begin{align*}
    |\mu|^2=&(\slangle{J}
    U_J\omega|\otimes\slangle{K}\omega|)P_-
    (|U_J\omega\srangle{J}\otimes|\omega\srangle{K})=\Tr_J(\Tr_K(P_-|\omega\srangle{K}\slangle{K}\omega|))
    \\=&
    \frac{2(K-J)+1}{2K+1}\Tr_J\Phi^{-k}(|\omega\srangle{K}\slangle{K}\omega|)= \frac{2(K-J)+1}{2K+1}.
  \end{align*}
\end{proof}

If $k=2(K-J)\geq0$ we therefore have
\begin{equation}
\slangle{K}\phi|\Phi^k(|\psi\srangle{J}\slangle{J}\psi|)|\phi\srangle{K}=
\frac{2J+1}{2K+1}\|\slangle{J} U_J\psi\|\phi\srangle{K}\|_{K-J}^2.
\end{equation}
If we introduce the channel 
\begin{equation}\label{eq:tildePhi}
  \widetilde\Phi^k(\rho)=\Phi^k(U^{\rlap{\phantom{-1}}}_J\rho
U_J^{-1})
\end{equation}
we see that
$$
\slangle{K}\phi|\widetilde\Phi^k(|\psi\srangle{J}\slangle{J}\psi|)|\phi\srangle{K}=
\frac{2J+1}{2K+1}\|\slangle{J}\psi\|\phi\srangle{K}\|_{K-J}^2
$$
or equivalently
\begin{equation}
  \widetilde\Phi^k(\rho)=\frac{2J+1}{2K+1}P_K(I_{K-J}\otimes
  \rho)P_K,\label{eq:tildeE-}
\end{equation}
where $P_K$ is the projection onto $\cH_K$ in $\cH_{K-J}\otimes\cH_J$.
In particular, $\widetilde\Phi^0$ is the identity map. 

If $K\leq J$ the corresponding result is that 
\begin{equation}
P_-\psi\otimes\phi= \mu' \
\slangle{K} U_K\phi\|\psi\srangle{J},\quad
|\mu'|^2=\frac{2(J-K)+1}{2J+1}.
\end{equation}
In particular, in this case
\begin{equation}
\slangle{K}\phi|\Phi^{-|k|}(|\psi\srangle{J}\slangle{J}\psi|)|\phi\srangle{K}=
\|\slangle{K} U_K\phi\|\psi\srangle{J}\|_{J-K}^2.
\end{equation}
Hence if $K\leq J$ and we now set $\widetilde\Phi^{-|k|}(\rho)=U_K^{-1}\Phi^k(\rho)
U_K$ we obtain
$$
\slangle{K}\phi|\widetilde\Phi^{-|k|}(|\psi\srangle{J}\slangle{J}\psi|)|\phi\srangle{K}=
\|\slangle{K}\phi\|\psi\srangle{J}\|_{J-K}^2.
$$

\section{The Bosonic formulation}\label{sec:bose}
The space $\cH_J$ may be identified with the completely symmetric
subspace $\bigotimes_{\text{sym}}^{2J}\cH_{1/2}$ of the tensor product $\bigotimes^{2J}\cH_{1/2}$. 

A particularly simple way to see this is to use the Schwinger 
representation of spin operators in terms of creation and annihilation
operators. Let $\cH_{1/2}$ be the one-particle space and let
$a_{\uparrow}^*$ and $a_{\downarrow}^*$ be the creation operators
corresponding to spin up and down respectively. They are the operators
which, for all positive integers $\ell$, map 
$\bigotimes_{\text{sym}}^{\ell}\cH_{1/2}$ to 
$\bigotimes_{\text{sym}}^{\ell+1}\cH_{1/2}$, such that for $\psi\in\bigotimes_{\text{sym}}^{\ell}\cH_{1/2}$
$$
a_{\uparrow}^*\psi=\sqrt{\ell+1}P_{\text{sym}}(|\uparrow\srangle{\frac12}\otimes \psi),
$$
and likewise for $a_{\downarrow}^*$, where $P_{\text{sym}}$ is the projection
onto the symmetric space $\bigotimes_{\text{sym}}^{\ell+1}\cH_{1/2}$. The
annihilation operators $a_{\uparrow}$, $a_{\downarrow}$ are the
adjoints of $a_{\uparrow}^*$ and $a_{\downarrow}^*$.

The symmetric subspace of $\bigotimes^{2J}\cH_{1/2}$ is the
subspace corresponding to $2J$ bosonic particles, i.e., the subspace
$$
a^*_\uparrow a_\uparrow+a^*_\downarrow a_\downarrow=2J.
$$
On this $2J+1$-dimensional subspace we observe that the operators 
\begin{eqnarray*}
S_x&=&\frac12(a^*_\uparrow a_\downarrow+a^*_\downarrow a_\uparrow)\\
S_y&=&\frac1{2i}(a^*_\uparrow a_\downarrow-a^*_\downarrow a_\uparrow)\\
S_z&=&\frac12(a^*_\uparrow a_\uparrow-a^*_\downarrow a_\downarrow)
\end{eqnarray*}
satisfy the correct commutation relations and
$S_x^2+S_y^2+S_z^2=J(J+1)$. 

The spin representation on $\cH_J$ may then be identified with the
space of $2J$ bosons over a $2$-dimensional one-particle space. 
In particular, the coherent state $|\omega\srangle{J}\in\cH_J$
is in the bosonic language the pure condensate wave function
$((2J)!)^{-1/2}(a_\omega^*)^{2J}|0\rangle$, where $|0\rangle$ is the
vacuum state (i.e., the state of zero particles) and 
$a_\omega^*$ is the creation of a particle in the state
$|\omega\rangle_{\frac12}$, i.e.,
$a_\omega^*=\slangle{\frac12}\uparrow|\omega\srangle{\frac12}a^*_\uparrow+
\slangle{\frac12}\downarrow|\omega\srangle{\frac12}a^*_\downarrow$.
We will use the canonical commutation relations that all creation
operators commute and 
$[a_{\omega'},a_\omega^*]=\slangle{\frac12}\omega'|\omega\srangle{\frac12}$.
This gives, in particular,
\begin{equation}\label{eq:commutation}
a_\uparrow a^*_\uparrow +a_\downarrow a^*_\downarrow=a^*_\uparrow
a_\uparrow+a^*_\downarrow a_\downarrow+2=2J+2
\end{equation}
on $\cH_J$.

The channel $\widetilde\Phi^k$ has a simple form using creation and
annihilation operators.
\begin{lm}[The channel $\widetilde\Phi^k$ in second quantization]\label{lm:Phi2nd}
If $\rho$ is a density matrix on
$\cH_J=\bigotimes_{\text{sym}}^{2J}\cH_{1/2}$ and $K=J+\frac{k}{2}$ for some
integer $k\geq0$ then
\begin{equation}\label{eq:kraus1}
\widetilde\Phi^k(\rho)=\frac{(2J+1)!}{(2K+1)!}\sum_{i_1,\ldots, i_{k}=\uparrow,\downarrow}a_{i_k}^*\cdots
a_{i_1}^*\rho a_{i_1}\cdots a_{i_k}.
\end{equation}
If $K=J-\frac{k}{2}$ with $k\geq0$ we have 
\begin{equation}\label{eq:kraus2}
\widetilde\Phi^{-k}(\rho)=\frac{(2K)!}{(2J)!}\sum_{i_1,\ldots, i_{k}=\uparrow,\downarrow}a_{i_k}\cdots
a_{i_1}\rho a_{i_1}^*\cdots a_{i_k}^*.
\end{equation}
\end{lm}
\begin{proof} We use the expression (\ref{eq:tildeE-}) for
  $\widetilde\Phi^k$. The space $\cH_K=\bigotimes_{\text{sym}}^{2K}\cH_{1/2}$ is
  the totally symmetric subspace of 
  $\left(\bigotimes_{\text{sym}}^{2(K-J)}\cH_{1/2}\right)\otimes\left(\bigotimes_{\text{sym}}^{2J}\cH_{1/2}\right)$.
  Thus by the definition of the creation and annihilation
  operators\footnote{The meaning of the operator
    $|\psi\rangle\otimes\rho\otimes\langle\psi|$ in
    (\ref{eq:tensordef}) is clear if $\rho$ is a rank one
    projection $|\phi\rangle\langle\phi|$ and it is defined in general by linearity}
  \begin{align}
  \widetilde\Phi^k(\rho)=&\frac{2J+1}{2K+1}\sum_{i_1,\ldots
    i_{k}=\uparrow,\downarrow}P_K|i_k\srangle{\frac12}\otimes\cdots\otimes|i_1\srangle{\frac12}\otimes\rho
  \otimes\slangle{\frac12} i_1|\otimes\cdots\otimes\slangle{\frac12} i_k|P_K\label{eq:tensordef}\\
  =&\frac{2J+1}{2K+1}\frac{(2J)!}{(2K)!}\sum_{i_1,\ldots
    i_{k}=\uparrow,\downarrow}a_{i_k}^*\cdots a_{i_1}^*\rho a_{i_1}\cdots a_{i_k}.
\end{align}
The case $K<J$ follows in the same way. 
\end{proof}
\begin{remark} Note that (\ref{eq:kraus1}) and (\ref{eq:kraus2}) are the Kraus representations of
  the completely positive trace preserving maps $\widetilde\Phi^k$.
\end{remark}
We shall use this to calculate the action of $\widetilde\Phi^k$ on coherent
states. 
\begin{lm}[The action of $\widetilde\Phi^k$ on coherent
  states]\label{lm:coherentoutput}
If $K=J+\frac{k}{2}$ for some
integer $k\geq0$ then
$\widetilde\Phi^k(|\uparrow\srangle{J}\slangle{J}\uparrow|)$
has the orthonormalized eigenfunctions
$$
\phi_j^{C,k}=(j!)^{-1/2}((2J+k-j)!)^{-1/2}(a^*_\downarrow)^{j}(a^*_\uparrow)^{2J+k-j}|0\rangle,\quad j=0,\ldots,2J+k=2K
$$
with corresponding eigenvalues 
$$
\lambda_j^{C,k}=\frac{2J+1}{2J+k+1}\frac{k!(2J+k-j)!}{(2J+k)!(k-j)!}
=\frac{2J+1}{2K+1}\frac{(2(K-J))!(2K-j)!}{(2K)!((2(K-J)-j)!},
$$
for $j=0,\ldots,k$ and zero for $j>k$. 
Note that the eigenvalues are listed in decreasing order. 
\end{lm}
\begin{proof}
This follows immediately from Lemma~\ref{lm:Phi2nd} since
\begin{align*}
\widetilde\Phi^k(|\uparrow\srangle{J}\slangle{J}\uparrow|)=&\frac{(2J+1)!}{(2J+k+1)!}\sum_{j=0}^k
{k\choose
  j}((2J)!)^{-1}(a^*_\downarrow)^{j}(a^*_\uparrow)^{2J+k-j}|0\rangle\langle
0|a_\uparrow^{2J+k-j}a_\downarrow^{j}.
\end{align*}
\end{proof}
An important ingredient in the proof of our main
Theorem~\ref{thm:Phimajor} below is to study the operator
$$
\Gamma^{C,k+1}_m=\sum_{j=0}^m a_\uparrow
|\phi_j^{C,k+1}\rangle\langle\phi_j^{C,k+1}|a_{\uparrow}^*+a_\downarrow
|\phi_j^{C,k+1}\rangle\langle\phi_j^{C,k+1}|a_\downarrow^*
$$
for $m\leq k$. Using the fact that
$a_\downarrow\phi_j^{C,k+1}=\sqrt{j}\,\phi_{j-1}^{C,k}$ for $j\geq1$ and
$a_\downarrow\phi_0^{C,k+1}=0$ and 
$a_\uparrow\phi_j^{C,k+1}=\sqrt{2J+k+1-j}\,\phi_{j}^{C,k}$ for
$j=0,\ldots,m$ we find that 
\begin{equation}\label{eq:Gammacoherent}
\Gamma^{C,k+1}_m=\sum_{j=0}^{m-1}
(2J+k+2)|\phi_j^{C,k}\rangle\langle\phi_j^{C,k}|
+(2J+k+1-m)|\phi_m^{C,k}\rangle\langle\phi_m^{C,k}|.
\end{equation}
\begin{remark} 
The expression in Lemma~\ref{lm:Phi2nd} for the channel
$\widetilde\Phi^k$ may be generalized to define analogous channels
between bosonic many-particle spaces where the one-particle space instead of
being 2-dimensional, as $\cH_{1/2}$, could be of arbitrary finite
dimension. We conjecture that Theorem~\ref{thm:Phimajor} holds also
in this case in the sense that pure condensates majorize these
channels. 
\end{remark}

\section{Proof of the main theorem for the channels $\Phi^k$}\label{sec:main}
\begin{proof}[Proof of Theorem~\ref{thm:Phimajor}]
If $k=2(K-J)\leq0$ then $\Phi^k(|\omega\rangle\langle\omega|)$ is rank
one and the result is obvious. 
We now consider the case $k=2(K-J)>0$. 
We first point out that from (\ref{eq:variational}) it is enough to
consider the rank one case, i.e., $\rho=|\psi\rangle\langle\psi|$ for
some $\psi\in\cH_J$. 
Since $I_{K-J}$ has rank $2(K-J)+1=k+1$ it is clear from (\ref{eq:tildeE-}) 
that $\Phi^k(|\psi\rangle\langle\psi|)$ has rank at most $k+1$. 
It is of course equivalent to consider the channel
$\widetilde\Phi^k$.
Let $\lambda_j^k(\psi)$, $j=0,1,\ldots,2J+k$ be the eigenvalues of 
$\widetilde\Phi^k(|\psi\rangle\langle\psi|)$ in decreasing order and counted with
multiplicity. Let $\phi^k_j$, $j=0,1,\ldots,2J+k$ be the corresponding
orthonormalized eigenvectors. For all $m\geq k$ we have 
$\sum_{j=0}^m\lambda^k_j(\psi)=\Tr\,\widetilde\Phi^k(|\psi\rangle\langle\psi|)=1$. 
The claim is that, moreover,
\begin{equation}\label{eq:msum}
\sum_{j=0}^m\lambda^k_j(\psi)\leq
\sum_{j=0}^m\lambda^{C,k}_j,
\end{equation}
for $m=0,\ldots,k-1$, where $\lambda_j^{C,k}$ are the eigenvalues
for the coherent states, which were given in Lemma~\ref{lm:coherentoutput}.

We shall prove (\ref{eq:msum}) by induction on $m$. For $m=0$ this is
easy since we clearly have from (\ref{eq:tildeE-}) that 
$
\lambda^k_0(\psi)\leq \frac{2J+1}{2K+1}=\lambda_0^{C,k}.
$
Let us now assume that we have proved (\ref{eq:msum}) for all integers
up to $m-1$ for some $m\geq1$ . We want to prove it for $m$. We shall
do this by induction on $k$. For $k\leq m$, (\ref{eq:msum}) is an
equality since both sides are 1. Let us assume, therefore, that we have
proved (\ref{eq:msum}) up to some $k\geq m$. We want to prove it for
$k+1$. 

Since $\phi^{k+1}_0,\ldots, \phi^{k+1}_m\in\cH_{J+(k+1)/2}$ are the
orthonormal eigenvectors corresponding to the $m$ top eigenvalues of
$\widetilde\Phi^{k+1}(|\psi\rangle\langle\psi|)$ we have
\begin{align}
  \sum_{j=0}^m\lambda^{k+1}_j(\psi)=&\frac{(2J+1)!}{(2J+k+2)!}\sum_{j=0}^m\sum_{i_1,\ldots,
    i_{k+1}=\uparrow,\downarrow}
  \langle\phi^{k+1}_j|
  a_{i_{k+1}}^*\cdots
  a_{i_1}^*|\psi\rangle\langle\psi| a_{i_1}\cdots
  a_{i_{k+1}}|\phi^{k+1}_j\rangle
  \nonumber\\
  =&\frac{(2J+1)!}{(2J+k+2)!}\sum_{i_1,\ldots,
    i_{k}=\uparrow,\downarrow}\Tr\left[\Gamma
    a_{i_k}^*\cdots
    a_{i_1}^*|\psi\rangle\langle\psi| a_{i_1}\cdots
    a_{i_{k}}\right]\nonumber\\=&\frac1{(2J+k+2)}\Tr[\Gamma
  \widetilde\Phi^k(|\psi\rangle\langle\psi|)]=
  \frac1{(2J+k+2)}\sum_{j=0}^{2J+k}\lambda^k_j(\psi)\langle\phi_j^k|\Gamma|\phi_j^k\rangle,
  \label{eq:iteration1}
\end{align}
where 
$$
\Gamma=\sum_{j=0}^m\left(a_\uparrow
|\phi^{k+1}_j\rangle\langle\phi^{k+1}_j|a_{\uparrow}^*+a_\downarrow
|\phi^{k+1}_j\rangle\langle\phi^{k+1}_j|a_\downarrow^*\right),
$$
is an operator on the space $\cH_{J+k/2}$.
Observe that since $\phi^{k+1}_j$ are $2J+k+1$ particle states we have
$$
\Tr\Gamma=\sum_{j=0}^m\langle\phi^{k+1}_j|a_{\uparrow}^*a_{\uparrow}+a_{\downarrow}^*a_{\downarrow}|\phi^{k+1}_j\rangle
=(m+1)(2J+k+1).
$$
Likewise, we have from (\ref{eq:commutation}) with $2J$ replaced by
$2J+k$, the operator inequalities 
\begin{align*}
0\leq\Gamma=&
a_{\uparrow}\sum_{j=0}^m|\phi^{k+1}_j\rangle\langle\phi^{k+1}_j|a_{\uparrow}^*+
a_{\downarrow}\sum_{j=0}^m|\phi^{k+1}_j\rangle\langle\phi^{k+1}_j|a_{\downarrow}^*\leq
a_\uparrow a_{\uparrow}^*+a_\downarrow a_\downarrow^*=(2J+k+2)I_{\cH_{J+k/2}}.
\end{align*}
We would therefore get an upper bound to the expression in
(\ref{eq:iteration1}) if $\Gamma$ is replaced by
\begin{equation}\label{eq:Gammaoptimizer}
\sum_{j=0}^{m-1}
(2J+k+2)|\phi_j^{k}\rangle\langle\phi_j^{k}|
+(2J+k+1-m)|\phi_m^{k}\rangle\langle\phi_m^{k}|.
\end{equation}
This gives the bound
\begin{align*}
\sum_{j=0}^m\lambda^{k+1}_j(\psi)\leq&
\frac{2J+k+1-m}{2J+k+2}\lambda^{k}_m(\psi)
+\sum_{j=0}^{m-1}\lambda^{k}_j(\psi)
\\=&\frac{2J+k+1-m}{2J+k+2}\sum_{j=0}^{m}\lambda^{k}_j(\psi)
+\frac{m+1}{2J+k+2}\sum_{j=0}^{m-1}\lambda^{k}_j(\psi).
\end{align*}
We conclude from the induction hypotheses on both $m$ and $k$ that 
$$
\sum_{j=0}^m\lambda^{k+1}_j(\psi)\leq\frac{2J+k+1-m}{2J+k+2}\sum_{j=0}^{m}\lambda^{C,k}_j
+\frac{m+1}{2J+k+2}\sum_{j=0}^{m-1}\lambda^{C,k}_j=\sum_{j=0}^m\lambda^{C,k+1}_j(\psi).
$$
That the last recursive identity holds for the coherent eigenvalues
follows since for coherent states we know from
(\ref{eq:Gammacoherent}) that $\Gamma$ is, in fact, equal to the
optimizing expression (\ref{eq:Gammaoptimizer}). It can also be seen
from the explicit formulas in Lemma~\ref{lm:coherentoutput}
(e.g. using induction).  The induction is thus complete.
\end{proof}
Note that as a special case we have seen in the proof that for $k=2(K-J)\geq0$
\begin{equation}\label{eq:normPhi-}
\|\Phi^k(\rho)\|\leq \frac{2J+1}{2K+1}.
\end{equation}
\section{The classical limit of the channels $\Phi^k$}\label{sec:classicallimit}
In this section we consider the limit as the dimension of the output
space tends to infinity, i.e., $K\to\infty$. 
\begin{lm}[The large $K$ limit of coherent state outputs]\label{lm:coherentlimit}
Let $f:[0,1]\to\R$ be a continuous function. Then
\begin{equation}\label{eq:cllimit}
\lim_{K\to\infty}
\frac{2J+1}{2K+1}\Tr_K\left[f\left(\frac{2K+1}{2J+1}\Phi^k(|\uparrow\srangle{J}\slangle{J}\uparrow|)\right)\right]
=\frac{2J+1}{4\pi}\int_{\rS^2}f(|\slangle{J}\omega|\uparrow\srangle{J}|^2)d\omega.
\end{equation}
\end{lm}
\begin{proof} We may of course replace the channel $\Phi^k$ by
  $\widetilde\Phi^k$. This a simple calculation based on the explicit
  expressions Lemma~\ref{lm:coherentoutput}. If $\theta$ denotes the
  polar angle of $\omega$, i.e., $\cos(\theta)$ is the $z$-component
  of $\omega$ then
  $|\slangle{J}\omega|\uparrow\srangle{J}|^2=\cos^{4J}(\theta/2)$.
  The integral over the sphere may hence be rewritten as
  $$
  \frac{2J+1}{4\pi}\int_{\rS^2}f(|\slangle{J}\omega|\uparrow\srangle{J}|^2)d\omega
  =(2J+1)\int_0^1f(t^{2J})dt.
  $$
  On the other hand the explicit eigenvalues in
  Lemma~\ref{lm:coherentoutput} give
  \begin{align*}
    \frac{2J+1}{2K+1}\Tr_K\left[f\left(\frac{2K+1}{2J+1}\Phi^k(|\uparrow\srangle{J}\slangle{J}\uparrow|)\right)\right]
    =&\frac{2J+1}{2K+1}\sum_{j=0}^{2(K-J)}f\left(\left(\frac{2K+1}{2J+1}\right)\lambda_j^C\right)
    \\=& \frac{2J+1}{2K+1}\sum_{j=0}^{2(K-J)}f\left(\frac{(2(K-J))!(2K-j)!}{(2K)!((2(K-J)-j)!}\right).
  \end{align*}
  It is an easy exercise which we leave to the reader to show that
  this converges to the above integral in the limit as $K\to\infty$. 
\end{proof}
It can be proved that the same limiting equality (\ref{eq:cllimit}) holds even  if 
$|\uparrow\srangle{J}\slangle{J}\uparrow|$ is replaced by any density
matrix $\rho$ on $\cH_J$, at least for a large class of
functions $f$. We will not do this here.  Instead, we shall restrict
ourselves to an inequality similar to (\ref{eq:cllimit}) for concave
functions $f$. To do this we shall 
use the Berezin-Lieb inequality (\ref{eq:BL}). 

\begin{lm}[The classical integral dominate the trace of concave
  functions of $\Phi$]\label{lm:berezinlieb}
Assume $f:[0,1]\to\R$ is a concave function. Then for any density
matrix $\rho$ on $\cH_J$ we have for all integers $k=2(K-J)\geq0$ 
$$
\frac{2J+1}{2K+1}\Tr_K\left[f\left(\frac{2K+1}{2J+1}\Phi^k(\rho)\right)\right]
\leq\frac{2J+1}{4\pi}\int_{\rS^2}f(\slangle{J}\omega|\rho|\omega\srangle{J})d\omega.
$$
\end{lm}
\begin{proof}
Again we consider the equivalent channel $\widetilde\Phi^k$.
The result follows from the Berezin-Lieb inequality (\ref{eq:BL}) if we can show that
the lower symbol of $\widetilde\Phi^k(\rho)$ satisfies
$$
\slangle{K}\omega|\widetilde\Phi^k(\rho)|\omega\srangle{K}=\frac{2J+1}{2K+1}
\slangle{J}\omega|\rho|\omega\srangle{J}.
$$
This is straightforward from (\ref{eq:tildeE-}).
\end{proof} 
We are now in a position to prove the main result Theorem~\ref{thm:conjecture}. 
In fact it is the analog of the main Theorem~\ref{thm:Phimajor} or
rather the equivalent formulation Corollary~\ref{cl:main} for the
classical map $\Phi^\infty$ from density matrices on $\cH_J$
to functions on (the classical phase space) $\rS^2$.
For classical functions the trace is replaced with the integral
over phase space. 
\begin{proof}[Proof of Theorem~\ref{thm:conjecture}]
From  Corollary~\ref{cl:main} and Lemma~\ref{lm:berezinlieb} we have
for all integers $k=2(K-J)\geq0$
\begin{align*}
\frac{2J+1}{2K+1}\Tr_K\left[f\left(\frac{2K+1}{2J+1}\Phi^k(|\uparrow\srangle{J}\slangle{J}\uparrow|)\right)\right]
\leq&\frac{2J+1}{2K+1}\Tr_K\left[f\left(\frac{2K+1}{2J+1}\Phi^k(\rho)\right)\right]\\
\leq\frac{2J+1}{4\pi}\int_{\rS^2}f(\slangle{J}\omega|\rho|\omega\srangle{J})d\omega.
\end{align*}
The result now follows from Lemma~\ref{lm:coherentlimit}.
\end{proof}

\appendix
\section{Proof of the generalized Wehrl conjecture,
  Theorem~\ref{thm:glauber}}

It is enough to prove Theorem~\ref{thm:glauber} for $n=1$. The general case follows by
induction as follows. Assume we have proved it for $n-1$. Then for
each $(p',q')\in\R^{n-1}$ we define an operator
$\widetilde\rho_{p',q'}$ on $L^2(\R)$ by
$$
\langle\phi|\widetilde\rho_{p',q'}|\psi\rangle=
\langle\phi|\otimes\langle p',q'|\rho|p',q'\rangle\otimes|\psi\rangle.
$$
Then 
$$
\rho_{p',q'}=(\Tr_{L^2(\R)}\widetilde\rho_{p',q'})^{-1}\widetilde\rho_{p',q'}
$$
is a density matrix on $L^2(\R)$ and we get from the inequality for
$n=1$ that 
\begin{align*}
  (2\pi)^{-n} \int_{\R^n\times \R^n}f(\rho^{\rm cl}(p,q) )dp dq
  \geq&(2\pi)^{-n} \int_{\R^n\times
    \R^n}f(\Tr_{L^2(\R)}\widetilde\rho_{p',q'}|\langle
  0,0|p_n,q_n\rangle|^2)dp dq.
\end{align*}
We have, however, that $\Tr_{L^2(\R)}\widetilde\rho_{p',q'}=\langle
p',q'|\Tr_n\rho|p',q'\rangle$, where $\Tr_n\rho$ is the density matrix
on $L^2(\R^{n-1})$ obtained by taking the partial trace on the $n$-th
variable. Thus, from the induction hypothesis, we find
\begin{align*}
(2\pi)^{-n} \int_{\R^n\times \R^n}f(\rho^{\rm
  cl}(p,q) )dp dq \geq&(2\pi)^{-n} \int_{\R^n\times \R^n}f(|\langle
0,0|p',q'\rangle|^2|\langle 0,0|p_n,q_n\rangle|^2)dpdq\\
=&(2\pi)^{-n} \int_{\R^n\times \R^n}f(|\langle
0,0|p,q\rangle|^2)dpdq.
\end{align*}

It remains to prove (\ref{eq:genwehrl}) in the case $n=1$.  We first
observe that, by possibly replacing $f(t)$ by $f(t)+at$, we can assume
that $f$ is non-negative. Moreover, using the monotone convergence
Theorem we can assume that $f$ is piecewise linear. In this case we
have the inequality
$$
|f(x)-f(y)| \leq C |x-y|
$$
for some $C>0$ and all $x,y\in [0,1]$.  Hence 
for all density matrices $\rho_1$, $\rho_2$ we have 
\begin{align*}
  \int_{\R^n\times \R^n}|f(\rho_1^{\rm
    cl}(p,q) )-f(\rho_2^{\rm
    cl}(p,q) )|dp dq \leq& C
  \int_{\R^n\times \R^n}|\rho_1^{\rm
    cl}(p,q) -\rho_2^{\rm
    cl}(p,q) |dp dq \\
  \leq&C\int_{\R^n\times \R^n}\langle p,q||\rho_1-\rho_2||p,q\rangle
  dp dq \\=&
  C\|\rho_1-\rho_2\|_1,
\end{align*}
where the the norm on the right is the trace norm. Hence it 
is enough to prove (\ref{eq:genwehrl}) for a subset of density
matrices that is dense in trace norm.

Let $|n\rangle$ denote the eigenfunctions of the harmonic oscillator 
$$
\left(-\frac{d^2}{dq^2}+q^2\right)|n\rangle=(2n+1)|n\rangle.
$$
We will prove (\ref{eq:genwehrl}) for the dense family of $\rho$ satisfying the property that
there exists a positive integer $N$ such that 
\begin{equation}\label{eq:rhoassumption}
\langle
n|\rho|m\rangle=0\quad\text{if }n>N \text{ or }m>N. 
\end{equation}
We shall apply the convenient complex notation, where
$z=2^{-1/2}(q+ip)$ and $\overline{z}=2^{-1/2}(q-ip)$. The Glauber coherent states may then be written
$$
|p,q\rangle=|z\rangle=\sum_{n=0}^\infty e^{-|z|^2/2}\frac{z^n}{\sqrt{n!}}|n\rangle.
$$
We identify the subspace $\cH_J$ with
span$\{|0\rangle,\ldots,|2J\rangle\}$ in $L^2(\R)$ in such a way that
$|M\srangle{J}=|M+J\rangle$ for $M=-J,\ldots,J$. Moreover, we also
identify the 3-sphere with the complex plane through stereographic
projection, such that the measure is $(1+|z|^2/4)^{-2}d^2z$.  With
these identification we can conveniently write the Bloch coherent
states (see \cite{L1}) as
$$
|z\srangle{J}=\sum_{n=0}^{2J}{2J\choose n}^{1/2}(1+|z|^2/4)^{-J}(\overline{z}/2)^n|n\rangle.
$$
It is now straightforward to see that if $\rho$ satisfies
assumption (\ref{eq:rhoassumption}) then 
$$
\langle z|\rho|z\rangle = \lim_{J\to\infty}\slangle{J} (2/J)^{\frac12}\,\overline{z}
|\,\rho\,|(2/J)^{\frac12}\overline{z}\srangle{J}.
$$
Using the fact that $(1+|z|^2/(2J))^{-2J}\leq (1+|z|^2/(2K))^{-2K}$
for all $z\in \C$ and all $J\geq K$, we easily see that for $J\geq N+2$
$$
\slangle{J}(2/J)^{\frac12}\overline{z}|\,\rho\,|(2/J)^{\frac12}\overline{z}\srangle{J}
\leq C_{\rho,N} (1+|z|^2/(2(N+2)))^{-2}
$$
if $\rho$ satisfies (\ref{eq:rhoassumption}). Since $f$ is
non-negative and bounded above by $t\mapsto at$ for some $a>0$ we
immediately find from dominated convergence that 
\begin{align}
  \int_\C f(\langle z|\rho|z\rangle)d^2z=&\lim_{J\to\infty}\int_\C f(\slangle{J}
  (2/J)^{\frac12}\overline{z}|\,\rho\,|(2/J)^{\frac12}\overline{z}\srangle{J})(1+|z|^2/(2J))^{-2}d^2z\nonumber\\
=&\lim_{J\to\infty}\frac{J}2\int_\C f(\slangle{J}
  z|\rho|z\srangle{J})(1+|z|^2/4)^{-2}d^2z.\label{eq:gaussianlimit}
\end{align}
Our main Theorem~\ref{thm:conjecture} and the observation that
$|0\rangle=|0\srangle{J}$ implies that 
$$
\int_\C f(\slangle{J}
  z|\rho|z\srangle{J})(1+|z|^2/4)^{-2}d^2z\geq\int_\C f(|\slangle{J}
  z|0\rangle|^2)(1+|z|^2/4)^{-2}d^2z.
$$
Since the density matrix $|0\rangle\langle0|$ clearly satisfies
(\ref{eq:rhoassumption}) we see from our main result and
(\ref{eq:gaussianlimit})  that (\ref{eq:genwehrl}) holds for all
$\rho$ satisfying (\ref{eq:rhoassumption}) and hence by approximation
for all density matrices on $L^2(\R)$. \hfill\qed

\end{document}